\documentclass[conference,letterpaper,10pt]{IEEEtran}
\usepackage[left=0.64in, right=0.63in, top=0.67in, bottom=1.04in]{geometry}
\IEEEoverridecommandlockouts

\usepackage[utf8]{inputenc} 
\usepackage[T1]{fontenc}
\usepackage{url}
\usepackage{ifthen}
\usepackage{cite}
\usepackage[cmex10]{amsmath} 
\usepackage{amssymb}
\usepackage{graphicx}
\usepackage{epstopdf}
\usepackage{epsfig}
\usepackage{multirow}
\usepackage{amsthm}
\usepackage{comment}
\usepackage{caption}
\usepackage{enumitem}
\usepackage{amsmath}
\usepackage{color}
\usepackage{bbm}
\usepackage{algorithm}
\usepackage{algpseudocode}
\usepackage{csquotes}
\usepackage{subfiles}
\usepackage{caption, multirow, makecell}
\usepackage{array}
\usepackage{caption}
\usepackage{nicematrix}
\usepackage{cuted}
\usepackage{subcaption}
\usepackage{mathtools}
\usepackage{blkarray}
\usepackage{enumitem}
\usepackage{bm}

\newtheorem{thm}{Theorem}
\newtheorem{lem}{Lemma}

\newtheorem{cor}{Corollary}

\newtheorem{defn}{Definition}

\newtheorem{exmp}{Example}

\allowdisplaybreaks

\def\BibTeX{{\rm B\kern-.05em{\sc i\kern-.025em b}\kern-.08em
    T\kern-.1667em\lower.7ex\hbox{E}\kern-.125emX}}

\begin{document}

\title{Wireless MapReduce Arrays for \\Coded Distributed Computing }

\author{\IEEEauthorblockN{~Elizabath Peter, K. K. Krishnan Namboodiri, and B. Sundar Rajan \\}
	\IEEEauthorblockA{Department of Electrical Communication Engineering, IISc
		Bangalore, India \\
		E-mail: \{elizabathp, krishnank, bsrajan\}@iisc.ac.in}
       \thanks{This work was supported partly by the Science and Engineering Research Board (SERB) of Department of Science and Technology (DST), Government of India, through J.C Bose National Fellowship to Prof. B. Sundar Rajan, and by the Ministry of Human Resource Development (MHRD), Government of India through Prime Minister's Research Fellowship to Elizabath Peter and K. K. Krishnan Namboodiri.}

}

\maketitle

\begin{abstract}
We consider a wireless distributed computing system based on the MapReduce framework, which consists of three phases: \textit{Map}, \textit{Shuffle}, and \textit{Reduce}. The system consists of a set of distributed nodes assigned to compute arbitrary output functions depending on a file library. The computation of the output functions is decomposed into Map and Reduce functions, and the Shuffle phase, which involves the data exchange, links the two. In our model, the Shuffle phase communication happens over a  full-duplex wireless interference channel. For this setting, a coded wireless MapReduce distributed computing scheme exists in the literature, achieving optimal performance under one-shot linear schemes. However, the scheme requires the number of input files to be very large, growing exponentially with the number of nodes. We present schemes that require the number of files to be in the order of  the number of nodes and achieve the same performance as the existing scheme. The schemes are obtained by designing a structure called wireless MapReduce array that succinctly represents all three phases in a single array. The wireless MapReduce arrays can also be obtained from the extended placement delivery arrays known for multi-antenna coded caching schemes.

\end{abstract}



\section{Introduction}
In the era of data-intensive applications, distributed computing systems are a necessity. 
In the popular MapReduce distributed computing framework \cite{DeG}, the computation is split into two stages: \textit{Map} and \textit{Reduce}. In the Map phase, the file library is distributed across a set of computing nodes to generate intermediate values according to the assigned map functions. These intermediate values are then exchanged among the nodes (called \textit{Shuffle} phase) to calculate the output functions distributedly using the Reduce functions. Coded MapReduce introduced in \cite{LMA,LMYA} utilizes the redundancy in the map computation phase to create coded multicasting opportunities in the Shuffle phase, thereby reducing the communication overhead or latency compared to uncoded data shuffling. The exact computation-communication tradeoff is characterized in \cite{LMYA}. The MapReduce distributed computing has been explored in several settings, such as considering different channel models in the Shuffle  phase \cite{LYMA,LCW}, identifying coded MapReduce schemes that require less number of files as input \cite{YTC,KoR,RaK,AgK,PLE}, designing Reduce functions that depend only on a subset of the files \cite{PRPA}, etc.

We study the wireless MapReduce distributed computing framework considered in \cite{LCW}, where the communication in the Shuffle phase happens over a full-duplex wireless interference channel. The MapReduce problem is defined by four parameters:  $K$, $Q$, $N$, and $r$, where $K$ is the number of computing nodes, $Q$ is the number of output functions, $N$ is the number of input files, and $r$ is the computation load in the Map phase denoting the average number of nodes that map a file. The computation of the $Q$ output functions is split into two stages: Map and Reduce. Each node computes a distinct set of output functions, and to perform the corresponding Reduce operations, intermediate values of all those functions on the $N$ files are required. Some of the intermediate values are available locally with each node from the Map phase, and the missing ones are obtained from the exchange in the Shuffle phase. 
The time required to transmit an intermediate value by a node, assuming no interference and successful decoding at the intended node, is taken to be $1$ unit of time. Then, the total time taken to exchange all the intermediate values in the Shuffle phase normalized by $NQ$ is called normalized delivery time. Under the constraint of one-shot linear transmissions in the Shuffle phase, the optimal normalized delivery time of a full-duplex wireless MapReduce distributed computing system is shown as $\frac{1-r/K}{\min\{2r,K\}}$, $r \in \{1,2,\ldots,K\}$ \cite{LCW}. A wireless MapReduce scheme achieving the above optimal performance is also given in \cite{LCW}. However, the achievable scheme in \cite{LCW} requires a large file library, i.e., $N$ has to be a multiple of $\binom{K}{r}$. Therefore, it is desirable to have schemes that achieve  optimal performance and need a smaller $N$. In this work, we present wireless distributed computing schemes having both of the above properties. Our contributions are summarized below.
\begin{itemize}
	\item We first introduce an array called wireless MapReduce array that represents the Map, Shuffle, and Reduce phases in a single array. We show that a solution for the full-duplex wireless MapReduce distributed computing problem is feasible from a wireless MapReduce array.
	\item Two constructions of wireless MapReduce array are presented, and the schemes resulting from them achieve the optimal normalized delivery time, under one-shot linear transmissions, with smaller $N$ values than the scheme in \cite{LCW}. The first construction is for the $r\geq K/2 $ setting, and the scheme obtained from it has $N=K$. The second construction is for $K=tr$, $t$ is a positive integer greater than one, and the resulting wireless distributed computing scheme has $N=K/r$. (Both the constructions assume $r \in \{1,\ldots,K\}$). Thus, our achievable schemes have $N$ in the order of $K$.
	\item We show that wireless MapReduce arrays can be obtained from the extended placement delivery arrays known for multi-antenna coded caching schemes \cite{NPR,YWCQC}. 
\end{itemize}
\noindent\textit{Notations:} The set of natural numbers is denoted by $\mathbb{N}$. For $m \in \mathbb{N}$, the set $\{1,\ldots,m\}$ is denoted as $[m]$.

\section{System Model}
\label{sec:sysmodel}
Consider a wireless MapReduce distributed computing system consisting of $K$ nodes, which are entitled to compute $Q$ output functions $\phi_1,\phi_2,\ldots,\phi_Q$ that depend on a set of $N$ files  $W=\{w_1,w_2,\ldots,w_N\}$, where $w_i \in \mathbb{F}_{2}^{F}, \forall i \in [N
]$. The $q^{th}$ output function is defined as  $\phi_q(W): (\mathbb{F}_{2}^{F})^N \rightarrow \mathbb{F}_{2}^{B}$, where $B \in \mathbb{N}$ and $ q \in [Q]$. The  output function $\phi_q(W)=\phi_q(w_1,w_2,\ldots,w_N)$, $q \in [Q]$, can be decomposed as:
  $\phi_q(w_1,\ldots,w_N)=h_q(g_{q,1}(w_1),\ldots,g_{q,n}(w_n)),$
where $g_{q,n}(w_n)$, $n \in [N]$, is called the Map function and is defined as $g_{q,n}(w_n): \mathbb{F}_{2}^F \rightarrow \mathbb{F}_{2}^T$, and $h_q: (\mathbb{F}_{2}^T)^N \rightarrow \mathbb{F}_{2}^B$ is called the Reduce function. The function $g_{q,n}(w_n)$ maps the input file $w_n$ to an intermediate value $v_{q,n} := g_{q,n}(w_n) \in \mathbb{F}_{2}^T$, and the Reduce function $h_q$ acts on the intermediate values $v_{q,n}$ of all the $N$ files associated with the $q^{th}$ output function to compute $\phi_q(w_1,\ldots,w_n)$.

Each node $k \in [K]$ is responsible to compute a distinct non-overlapping set of output functions $\mathcal{W}_k \subset \{\phi_1,\ldots, \phi_Q\}$. i.e., $Q \geq K $ and $\mathcal{W}_{k_1} \cap \mathcal{W}_{k_2} = \emptyset$ for $k_1 \neq k_2$. A symmetric task assignment is assumed across all the nodes, therefore,  $\frac{Q}{K} \in \mathbb{N}$ and $|\mathcal{W}_{1}|=\cdots=|\mathcal{W}_{K}|=\frac{Q}{K}$. 
The description of the Map, Shuffle, and Reduce phases are given below.  

\subsubsection{Map Phase}
Each node $k \in [K]$ is assigned a set of files $\mathcal{M}_k \subseteq W$ to compute the intermediate values $v_{q,n}, \forall q \in [Q]$ (of all the $Q$ output functions), corresponding to the files in $ \mathcal{M}_k$. That is,  node $k$ computes $v_{1,n}, v_{2,n},\ldots, v_{Q,n}$,
 $\forall w_n \in \mathcal{M}_k$. Note that every file should be mapped by at least one node and $\cup_{k = 1,\ldots,K}\mathcal{M}_k=\{w_1,\ldots,w_N\}$. The average number of nodes that map a file is called computation load and is denoted by $r$. i.e., $r \triangleq \frac{\sum_{k=1}^{K}|\mathcal{M}_k|}{N}$.
 
 \subsubsection{Shuffle Phase}
 Following the Map phase, the nodes exchange the intermediate values that are available locally to facilitate the computation of the output functions. Node $k \in [K]$ requires the set of intermediate values $\{v_{q,n}, \forall n \in [N]: \phi_q \in \mathcal{W}_k\}$ to evaluate the output functions in $\mathcal{W}_k$. The following set of intermediate values  $\{v_{q,n}, \forall q \in [Q]: w_n \in \mathcal{M}_k \}$ is known to node $k$ from the Map phase, in which $\{v_{q,n}: \phi_q \in \mathcal{W}_k \text{ and } w_n \in \mathcal{M}_k\}$ is the set of intermediate values useful for node $k$ to compute the output functions in $\mathcal{W}_k$. Then, the remaining set of intermediate values required by node $k$, $\{v_{q,n}: \phi_q \in \mathcal{W}_k \text{ and } w_n \notin \mathcal{M}_k\}$, is obtained from other nodes through the communication among them. The communication between the nodes, in this phase, happens over a full-duplex wireless interference channel. It is assumed that channel state information is available to all the nodes. Let $S$ be the number of channel uses needed to provide the missing intermediate values required for all the nodes. It is assumed that a pre-agreed protocol exists among the nodes for transmission in the shuffle phase. The set of users transmitting in the $s^{th}$ channel use is denoted as $\mathcal{U}_s$, where $\mathcal{U}_s\subseteq [K]$ and $|\mathcal{U}_s|\leq K$. The message received by node $k$ corresponding to the communication in the $s^{th}$ channel use is denoted as $y_k(s)$, and is given as:
 $  y_k(s) = \sum_{m \in \mathcal{U}_s}h_{k,m}{x_m}(s)+z_k(s),$
 where $h_{k,m} \in \mathbb{C}$ is the channel gain from the transmitting node $m$ to the receiving node $k$, $x_m(s) \in \mathbb{C}$ is the message sent by node $m$ in the $s^{th}$ channel use satisfying the power constraint $\mathbb{E}[|x_m(s)|^2]\leq P$, and $z_k(s)$ is the zero-mean unit variance circularly symmetric additive white complex gaussian noise observed at the receiving node $k$ (note that, $T \approx \log P + o(\log P)$). The intermediate values $v_{q,n} \in \mathbb{F}_{2}^T$ are encoded into $\tilde{v}_{q,n}\in \mathbb{C}$ and these encoded intermediate values are used for transmissions in the Shuffle phase. We assume a one-shot linear communication in the Shuffle phase. Also, we neglect the additive noise $z_k(s)$ in the further analysis assuming the transmit power is large enough to keep the signal-to-noise ratio high.

  Consider a transmission in one channel use, say $s \in [S]$.  Let the set of coded intermediate values transmitted in the $s^{th}$ channel use be denoted as $\mathcal{N}_s$, and the set of intended receiving nodes be denoted as $\mathcal{R}_s$, where $ \mathcal{R}_s \subseteq [K]$. From the transmission made by all the users in $\mathcal{U}_s$, each node in $\mathcal{R}_s$ obtains one coded intermediate value and it will not be repeated in another channel use. Then, $|\mathcal{N}_s|=|\mathcal{R}_s|$ and $\mathcal{N}_s \cap \mathcal{N}_{s^{\prime}}=\emptyset,$ $\forall s \neq s^{\prime}$ and $s,s^{\prime}\in [S]$.  Consider a node $m\in \mathcal{U}_s$. Then, the transmitted message by node $m$ in the $s^{th}$ channel use, denoted as $x_{m}(s)$, is of the form: 
 \begin{equation*}
   x_{m}(s) = \sum_{(q,n): \tilde{v}_{q,n} \in \mathcal{N}_s, n \in \mathcal{M}_m} \alpha_{m,q,n}\tilde{v}_{q,n},
 \end{equation*}
 where $\alpha_{m,q,n} \in \mathbb{C}$ is the precoding coefficient. The intended nodes in $\mathcal{R}_s$ will be able to decode its desired $\tilde{v}_{q,n}$ from the linear combination of the messages received in the $s^{th}$ channel use. For any node $k \in \mathcal{R}_s$, the received message $y_k(s)$ consists of an intermediate value desired by node $k$ and other unwanted intermediate values, in which some are available locally and the rest need to be zero-forced by  properly designing the precoding coefficients. Upon receiving the required intermediate values, the nodes convert the coded intermediate values $\tilde{v}_{q,n}$ back to $v_{q,n} \in \mathbb{F}_2^{T}$ before the Reduce phase.
 
 \subsubsection{Reduce Phase}
 Each node $k \in [K]$ is assigned to compute a subset of output functions $\mathcal{W}_k$, where $|\mathcal{W}_k|=\frac{Q}{K} \in \mathbb{N}$. 
 From the Map phase and the communication in the Shuffle phase, each node $k \in [K]$ obtains the required intermediate values to compute the corresponding Reduce functions, i.e., $h_q(v_{q,1},\ldots,v_{q,N})$, for all $ \phi_q \in \mathcal{W}_k$.
 
The performance measure of the communication in the Shuffle phase is called normalized delivery time, denoted by $L$, and is defined as:
$
    L = \displaystyle\lim_{P \rightarrow \infty}\lim_{T \rightarrow \infty}\frac{S}{NQT/\log P}.$
 
 The performance of the wireless distributed computing system is characterized in terms of computation load, $r$ and normalized delivery time, $L$. The computation load-normalized delivery time pair $(r,L)$ is said to be achievable if there exists a wireless MapReduce scheme (with one-shot linear assumption in the Shuffle phase) that could decode all the intermediate values with vanishing error probability as $T$ increases. 
Our objective is to design wireless MapReduce schemes that exist for small  number of files, $N$, and  achieve the optimal normalized delivery time under one-shot linear schemes for a computation load, $r$. The communication in the Shuffle phase is said to be one-shot if no intermediate value is transmitted in more than one channel use. Consequently, the nodes can decode any required intermediate value from at most one channel use.
 \section{Main Results}
 The existing wireless MapReduce schemes require $N$ to be sufficiently large (i.e., $N$ needs to be a multiple of $\binom{K}{r}$). In this section, we present wireless MapReduce schemes that exist when the number of files $N$ is linear with respect to $K$. To describe our schemes, we define an array called Wireless MapReduce Array that succinctly represents all the three phases in a single array. 
 
 \begin{defn}[Wireless MapReduce Array] For positive integers $K$, $N$, $S$, and $r$, an $N \times K$ array $\mathbf{A} = [a_{i,k}]$, $i \in [N]$, $k \in [K]$, composed of `$\star$' and integers from $[S]$ is said to be a $(K,N,r, S)$ wireless MapReduce array if the following conditions are satisfied.\\
 A1. Each row has $r$ $\star$'s. \\
 A2. Each integer in $[S]$ appears $\min\{2r,K\}$ times in the array and not more than once in any column.\\
 A3. Any row in the subarray $\mathbf{A}^{(s)}$, obtained by eliminating all the rows and columns not containing the integer $s$, has at most $r$ integers.
 \end{defn}

\begin{exmp}
	   $(5,5,3,2)$ wireless MapReduce array $\mathbf{A}:$
	 \begin{equation*}
	 \mathbf{A} = 
	 \begin{bmatrix}
        \star & 1 & 1 & \star & \star \\
        \star & \star & 2 & 1 & \star \\
        \star & \star & \star & 2 & 1 \\
        1 & \star & \star & \star & 2 \\
        2 & 2 & \star & \star & \star 
	 \end{bmatrix}.
	 \label{eq:ex_wmra}
	 \end{equation*}
	 It is evident that the array $\mathbf{A}$ satisfies  conditions A1 and A2. To verify condition $A3$, consider the subarrays $\mathbf{A}^{(1)}$ and $\mathbf{A}^{(2)}$ shown below.
	 \begin{equation*}
	   \mathbf{A}^{(1)} = 
	   \begin{bmatrix}
	      \star & 1 & 1 & \star & \star \\
	      \star & \star & 2 & 1 & \star \\
	      \star & \star &  \star & 2 & 1 \\
	      1 & \star & \star & \star & 2
	   \end{bmatrix} \quad \quad
	    \mathbf{A}^{(2)} = 
	   \begin{bmatrix}
	   \star & \star & 2 & 1 & \star \\
	   \star & \star &  \star & 2 & 1 \\
	   1 & \star & \star & \star & 2 \\
	   2 & 2 & \star & \star & \star
	   \end{bmatrix}.
	 \end{equation*}
	 No row in $\mathbf{A}^{(1)}$ and $\mathbf{A}^{(2)}$ contains more than $3$ integers. Hence, condition A3 is satisfied.
\end{exmp}

  Let $g=\min\{2r,K\}$ denote the number of times each integer appears in $\mathbf{A}$. Then, the total number of integers present in $\mathbf{A}$ is $Sg$. Alternatively, we count the number of integers present in $\mathbf{A}$ row-wise and obtain it as $N(K-r)$. Thus, we get $ N(K-r) = Sg$, which implies $S= \frac{N(K-r)}{g}$.
   A $(K,N,r,S)$ wireless MapReduce array leads to a wirelss coded distributed computing scheme based on the MapReduce framework for a system with  $K$ nodes, $Q$ output functions, and $N$ files as described below.

The $K$ columns of the wireless MapReduce array $\mathbf{A}$ correspond to each computing node, and the $N$ rows of  $\mathbf{A}$ correspond to the files. Without loss of generality, we assume that $\frac{Q}{K}=1$. When $\frac{Q}{K} \in \mathbb{N}\backslash \{1\}$, the transmission policy in the Shuffle phase needs to be repeated $Q/K$ times.

\textit{Map Phase:} The set of files mapped by each node, $\mathcal{M}_k$, is given as:
\begin{equation}
 \mathcal{M}_k = \{w_i: a_{i,k}= \star \}.
\end{equation}
For each mapped file  $w_n \in \mathcal{M}_k$, the node computes the intermediate value $v_{q,n}$ of all the $Q$ output functions. Hence, $a_{i,k}=\star$ implies that the node $k$ has the intermediate values $v_{q,i}$, $\forall q \in [Q]$.  The computation load in the Map phase is obtained as: $\frac{\sum_{k=1}^{K}|\mathcal{M}_k|}{N}=Nr/N =r$ (from condition $A1$).

\textit{Shuffle Phase:} In the Shuffle phase, the nodes exchange the intermediate values among themselves such that all the nodes obtain the remaining intermediate values that are required to calculate the output functions in $\mathcal{W}_k$. The exchange of intermediate values are completed in $S$ channel uses, and a set of nodes act as both transmitters and receivers in a channel use. The transmissions are described below.

Consider an integer $s \in [S]$ in $\mathbf{A}$, and $s$ appears $g=\min\{2r,K\}$ times in $\mathbf{A}$. Let $a_{i_1,k_1}=a_{i_2,k_2}=\cdots = a_{i_{g},k_{g}}=s$. Then, the nodes $k_1,k_2,\ldots,k_{g}$ transmit and receive simultaneously in the $s^{th}$ channel use. That is, $\mathcal{U}_s =\{k_1,k_2,\ldots,k_{g}\}= \mathcal{R}_s$. The set of coded intermediate values transmitted in the $s^{th}$ channel use, $\mathcal{N}_s$, correspond to the $\star$ entries in $\mathbf{A}^{(s)}$. That is,  $\mathcal{N}_s= \{\tilde{v}_{k_j,i_n}: a_{i_n,k_j}=\star, n \in [g], j \in [g]\}$. The transmission made by node $k_j$, $j \in [g]$, is: 
\begin{equation}
  x_{k_j}(s) = \sum_{(q,n):\tilde{v}_{q,n} \in \mathcal{N}_s, a_{n,k_j}=\star  } \alpha_{k_j,q,n}\tilde{v}_{q,n}.
\end{equation}
From the transmissions in the $s^{th}$ channel use, each node in $\mathcal{R}_s$ is able to obtain one missing intermediate value that is required to compute its assigned Reduce function. Consider a node $k \in \mathcal{R}_s$. The received message at node $k$ takes the form: 
\begin{subequations}
\begin{align}
  y_k(s) & = \sum_{m \in \mathcal{U}_s}h_{k,m}x_{m}(s) + z_k(s) \label{eq:rxmsg1}\\
  & = \sum_{m \in \mathcal{U}_s} \sum_{(q,n):\tilde{v}_{q,n} \in \mathcal{N}_s, a_{n,m}=\star  } h_{k,m} \alpha_{m,q,n}\tilde{v}_{q,n}. 
  \label{eq:rxmsg2}
\end{align}
\end{subequations}
Assuming high SNR regime, we neglect the additive noise in \eqref{eq:rxmsg1} in the further analysis. It is assumed that the channel coefficients are known to all the users. Note that any intermediate value $v_{q,n}$ is available with $r$ nodes, and any row in $\mathbf{A}^{(s)}$ contains only at most $r$ integers. 
Therefore, proper design of precoding coefficients enables the nodes in $\mathcal{R}_s$ to decode the desired intermediate values from the transmissions. After the transmissions in $S$ channel uses, each node obtains the missing intermediate values required to compute the output functions. Note that the retrieved intermediate values are in coded form, and they need to be decoded back to uncoded form before the Reduce phase.

\textit{Reduce Phase:} Without loss of generality, we assume  ${Q}={K}$ and the node $k \in [K]$ is assigned to compute the output function $\phi_k(w_1,w_2,\ldots,w_N)$. Since $\phi_k(w_1,w_2,\ldots,w_N) = h_k (v_{k,1},v_{k,2},\ldots,v_{k,N})$, node $k$ computes the Reduce function $h_k$ using the intermediate values of all the $N$ files of the $k^{th}$ output function. Thus, all the nodes are able to compute their assigned output functions.

From the above analysis, we have the following lemma.

\begin{lem}
	Given a $(K,N,r,S)$ wireless MapReduce array, it is possible to obtain a wireless MapReduce scheme for a system having $K$ users, $Q=K$ output functions, and $N$ files with a computation load $r$ and a normalized delivery time $L = S/NQ= \frac{1}{\min\{2r,K\}}\left(1-\frac{r}{K} \right)$.
	\label{lem1}
\end{lem}

\subsection{Constructions of Wireless MapReduce Arrays}
We present two constructions of wireless MapReduce arrays that result in coded distributed computing schemes having $N$ in the order of $K$. Also, the schemes achieve the optimal normalized delivery time under the constraint of one-shot linear transmissions. Thus, we have the following theorem.

\begin{thm}
	For a wireless distributed computing system having $K$ nodes, $Q=K$ output functions, and $N$ files, the optimal normalized delivery time,  under one-shot linear schemes,
	\begin{equation}
	L(r)= \frac{1}{\min\{2r,K\}}\left(1-\frac{r}{K}\right), 
	\label{eq:load}
	\end{equation}
for a computation load $r \in [K]$, is achievable in the following cases:
	\begin{enumerate}[label=(\alph*)]
		\item $r \geq \frac{K}{2}$,
		\item $K = tr$, $t \in \mathbb{N} \backslash \{1\}$,
	\end{enumerate}
with $N=K$ and $N=K/r$, respectively.
\label{thm1}
\end{thm}

\begin{proof}
	The schemes achieving the normalized delivery time in \eqref{eq:load} are obtained by designing appropriate wireless MapReduce arrays for both cases.

	We first consider case (a) $r \geq K/2 $. A $(K,K,r,K-r)$ wireless MapReduce array $\mathbf{A}=[a_{i,k}]$ is constructed as described below. The position of $\star$'s in the $k^{th}$ column of $\mathbf{A}$ is as follows:
	\begin{align}
	  a_{i,k}= \star,\text{ } \forall (i,k) \in [K]\times [K]\text{ such that } \langle i-k \rangle_K < r.
	  \label{eq:starplace}
	\end{align}
	From \eqref{eq:starplace}, it easy to see that there are $r$ consecutive $\star$'s in a row and $r$ consecutive $\star$'s in a column. Thus, condition $A1$ is satisfied. In the $k^{th}$ column, $a_{k,k}=a_{\langle k+1 \rangle_K,k}=\cdots=a_{\langle k+r-1 \rangle_K,k}=\star$. Next, we assign the integers. Consider an integer $s \in [K-r]$. Then, $\forall k \in [K]$,
	\begin{equation}
	 a_{i,k} = s \text{ for } i=\langle k-1+r+s \rangle_K.
	 \label{eq:integerplace}
	\end{equation}
From \eqref{eq:integerplace}, it is evident that $s$ appears in all the $K$ columns and does not appear more than once in any column. Thus, condition $A2$ is also satisfied. Now, it remains to verify condition $A3$. Note that all the $K-r$ integers appear in all the $K$ columns. Therefore, $\mathbf{A}^{(s)}=\mathbf{A}$, $\forall s\in [K-r]$, and each row of $\mathbf{A}$ has $K-r$ integers (from condition $A1$). Since $r \geq \frac{K}{2}$, we have $K-r \leq r$. Thus, condition $A3$ is satisfied, and $\mathbf{A}$ is a $(K,K,r,K-r)$ wireless MapReduce array. 

Using  Lemma~\ref{lem1} and the above constructed $(K,K,r,K-r)$ wireless MapReduce array $\mathbf{A}$,  we obtain a distributed computing scheme for a full-duplex wireless distributed computing system with $K$ users, $Q=K$ output functions, and $N=K$ files. Since $K \leq 2r$, we get $g=K$. Thus, for a computation load $r$, the scheme achieves the normalized delivery time  $L = \frac{1}{K}\left(1-\frac{r}{K} \right)$, which is optimal under one-shot linear schemes \cite{LCW}.

Next, we consider case (b) $ K= tr$, where $t \in \mathbb{N}\backslash \{1\}$. In this case, we first construct a $(\frac{K}{r},\frac{K}{r},1,\frac{t(t-1)}{2})$ wireless MapReduce array $\mathbf{B} = [b_{i,k}]$ as given below. The $\star$'s appear in $\mathbf{B}$ as:
\begin{equation}
 b_{i,i} = \star, \forall i \in [K/r].
 \label{eq:starplace2}
\end{equation}
Every row and column of $\mathbf{B}$ contains only one $\star$, thus satisfying condition $A1$. Each integer $s \in [\frac{t(t-1)}{2}]$ appears twice in $\mathbf{B}$, and the subarray $\mathbf{B}^{(s)}$ is always of size $2 \times 2$, $\forall s \in [\frac{t(t-1)}{2}]$. Let $s$ be expressed as $s=\frac{u(u-1)}{2}+v$, where $u \in [t-1]$ and $v \in [u]$. Then, $s$ appears in $\mathbf{B}$ as follows: 
		\begin{align}
		  b_{v,t-u+v} =s \text{ and }
		  b_{t-u+v,  v }=s.
		  \label{eq:integerplace2}
		\end{align}
From \eqref{eq:starplace2}, we have $b_{i,i}=\star$, where $i \in \{v, t-u+v \}$. Therefore,   $\mathbf{B}^{(s)}$ is of the form $\mathbf{B}^{(s)}=\begin{bmatrix}
\star & s\\s & \star
\end{bmatrix}$ or $\begin{bmatrix}
s & \star\\\star & s
\end{bmatrix}$, thus ensuring conditions $A2$ and $A3$. Upon constructing $\mathbf{B}$, we concatenate it horizontally $r$ times to obtain another array $\mathbf{C}$. That is, $\mathbf{C}=\underbrace{[\mathbf{B}; \cdots;\mathbf{B}]}_\text{$r$ times}$, and $\mathbf{C}$ is, indeed, a $(K,\frac{K}{r},r,\frac{t(t-1)}{2})$ wireless MapReduce array. Condition $A2$ follows directly from $\mathbf{B}$. There are $r$ $\star$'s in every row  of $\mathbf{C}$ as $\mathbf{B}$ contains only one $\star$ in every row, thus condition $A1$ is satisfied. In a similar manner, condition $A3$ also holds because $\mathbf{C}^{(s)}=\underbrace{[\mathbf{B}^{(s)}; \cdots;\mathbf{B}^{(s)}]}_\text{$r$ times}$, $\forall s \in [\frac{t(t-1)}{2}]$. Thus, each row of $\mathbf{C}^{(s)}$ contains exactly $r$ integers. 

We use $\mathbf{C}$ to obtain a wireless MapReduce distributed computing scheme for a system with $K$ users, $Q=K$ output functions, and $N=\frac{K}{r}$ files. The resulting scheme from $\mathbf{C}$ achieves the optimal normalized delivery time (under one-shot linear schemes) $L =  \frac{1}{2r}\left(1-\frac{r}{K} \right)$ for the computation load  $r$.


This completes the proof of Theorem~\ref{thm1}.
\end{proof}
We present an example corresponding to each case to describe the wireless distributed computing schemes. 
\begin{exmp}
Consider a full-duplex wireless distributed computing system with $K=Q=5$, $N=5$, and $r=3$. 
\label{exmpcasea}
\end{exmp}
Since $r \geq K/2$, we construct a $(5,5,3,2)$ wireless MapReduce array $\mathbf{A}$ as described in case (a). Thus, 
 \begin{equation}
 \small
\mathbf{A} = 
\begin{bmatrix}
\star & 2 & 1 & \star & \star \\
\star & \star & 2 & 1 & \star \\
\star & \star & \star & 2 & 1 \\
1 & \star & \star & \star & 2 \\
2 & 1 & \star & \star & \star 
\end{bmatrix}.
\label{eq:ex_casea}
\end{equation}
The set of files mapped by each node is represented by $\star$'s in its corresponding column. For instance, $\mathcal{M}_1=\{w_1,w_2,w_3\}$. Each node $k \in [5]$ generates the intermediate values of all the $5$ output functions associated with the files in $\mathcal{M}_k$. Assume that the $k^{th}$ node is entitled to compute the $k^{th}$ Reduce function $\phi_k$. To compute the assigned Reduce function, each node needs five intermediate values, out of which three intermediate values are already available with the node from the Map phase. The remaining two intermediate values are obtained from the data exchange among $5$ nodes. The Shuffle phase gets completed in $2$ channel uses. 
All the $5$ nodes transmit and receive in both the channel uses. The set of coded intermediate values transmitted in the two channel uses are: $\mathcal{N}_1=\{\tilde{v}_{1,4}, \tilde{v}_{2,5}, \tilde{v}_{3,1}, \tilde{v}_{4,2}, \tilde{v}_{5,3}, \}$ and $\mathcal{N}_2=\{\tilde{v}_{1,5}, \tilde{v}_{2,1}, \tilde{v}_{3,2}, \tilde{v}_{4,3}, \tilde{v}_{5,4}, \}$. In the interest of space, the transmissions in the first channel use are only given: $x_1(1)=\alpha_{1,3,1}\tilde{v}_{3,1}+\alpha_{1,4,2}\tilde{v}_{4,2}+\alpha_{1,5,3}\tilde{v}_{5,3}$,
$x_2(1)=\alpha_{2,1,4}\tilde{v}_{1,4}+\alpha_{2,4,2}\tilde{v}_{4,2}+\alpha_{2,5,3}\tilde{v}_{5,3}$, $x_3(1)=\alpha_{3,1,4}\tilde{v}_{1,4}+\alpha_{3,2,5}\tilde{v}_{2,5}+\alpha_{3,5,3}\tilde{v}_{5,3}$,
$x_4(1)=\alpha_{4,1,4}\tilde{v}_{1,4}+\alpha_{4,2,5}\tilde{v}_{2,5}+\alpha_{4,3,1}\tilde{v}_{3,1}$,
$x_5(1)=\alpha_{5,2,5}\tilde{v}_{2,5}+\alpha_{5,3,1}\tilde{v}_{3,1}+\alpha_{5,4,2}\tilde{v}_{4,2}$.
From the transmissions in the first channel use, each user gets one missing intermediate value. Consider user $1$. Corresponding to the transmissions in the first channel use, the received message at user $1$ is of the form: $y_1(1)=\sum_{m \in [5]}h_{1,m}x_m(1)+z_1(1)$. Assuming high-SNR regime, $z_1(1)$ is neglected and $y_1(1)$ can be rewritten as follows:
\begin{align}
  & y_1(1)  = \mathbf{h}_{1,\{234\}}^{\mathrm{T}}\bm{\alpha}_{\{234\}1,4}\tilde{v}_{1,4}+\mathbf{h}_{1,\{345\}}^{\mathrm{T}}\bm{\alpha}_{\{345\}2,5}\tilde{v}_{2,5}+ \notag\\
   & \mathbf{h}_{1,\{145\}}^{\mathrm{T}}\bm{\alpha}_{\{145\}3,1}\tilde{v}_{3,1}+  \mathbf{h}_{1,\{125\}}^{\mathrm{T}}\bm{\alpha}_{\{125\}4,2}\tilde{v}_{4,2}+ \notag \\
 & \mathbf{h}_{1,\{123\}}^{\mathrm{T}}\bm{\alpha}_{\{123\}5,3}\tilde{v}_{5,3},
 \label{eq:rx}
\end{align}
where $\mathbf{h}_{1,\{m_1m_2m_3\}}=[h_{1,m_1}, h_{1,m_2}, h_{1,m_3}]^{\mathrm{T}} \in \mathbb{C}^{r \times 1}$ represents the channel gains between node 1's receiving antenna and the transmitting antennas of nodes $m_1$, $m_2$, $m_3$, $\bm{\alpha}_{\{m_1m_2m_3\},q,n}=[\alpha_{m_1,q,n},\alpha_{m_2,q,n},\alpha_{m_3,q,n},] \in \mathbb{C}^{r \times 1}$, $q \in [Q]$, $n \in [N]$, represents the precoding vector of the coded intermediate value $\tilde{v}_{q,n}$. The channel gains, precoding coefficients are known to all the users, and the coded intermediate values $\tilde{v}_{3,1}$, $\tilde{v}_{4,2}$, $\tilde{v}_{5,3}$ are available to node $1$ from the Map phase. The coded intermediate value $\tilde{v}_{2,5}$ is not available with node $1$, hence it needs to be eliminated from \eqref{eq:rx} to enable the decoding of $\tilde{v}_{1,4}$. This elimination is facilitated by designing $\bm{\alpha}_{\{345\}2,5} \perp \mathbf{h}_{1,\{345\}} $, and such a design of the precoding vector is possible. Thus, node $1$ can decode $\tilde{v}_{1,4}$. Similarly, from $y_1(2)$, node $1$ gets $\tilde{v}_{1,5}$. In a similar manner, all other nodes decode their desired intermediate values from the received messages. Thus, $L=2/25$. From the transmissions and the decoding criterion, it is evident that communication in the Shuffle phase is one-shot and linear. Therefore, the obtained normalized delivery time $L=2/25$ is optimal, under one-shot linear schemes, for the computation load $r=3$.

\textit{Comparison with the scheme in \cite{LCW}}: For $K=Q=5$ and $r=3$, the scheme in \cite{LCW} requires $N$ to be $\binom{K}{r}=10$.

\begin{exmp}
 A full-duplex wireless distributed computing system with $K=Q=6$, $N=3$, and $r=2$.
\end{exmp}
For this example, we construct a $(6,3,2,3)$ wireless MapReduce array $\mathbf{C}$ as explained in case (b) ($t=3$). The array $\mathbf{C}$ is constructed from a $(3,3,1,3)$ wireless MapReduce array $\mathbf{B}$. The arrays $\mathbf{B}$ and $\mathbf{C}$ are given below.
\begin{equation*}
\mathbf{B}=
\begin{bmatrix}
 \star & 2 & 1\\
 2 & \star & 3\\
 1 & 3 &\star
 \end{bmatrix} \quad 
 \mathbf{C}=
 \begin{bmatrix}
 \star & 2 & 1 & \star & 2 & 1\\
 2 & \star & 3 & 2 & \star & 3\\
 1 & 3 &\star & 1 & 3 &\star
 \end{bmatrix} 
\end{equation*}
We omit further details due to the space constraints. Note that in each channel use, $4$ users transmit and receive simultaneously. The transmissions and the decoding procedure are similar to Example \ref{exmpcasea}. The normalized delivery time achieved by the wireless distributed computing scheme resulting from $\mathbf{C}$ is $L=1/6$, which is optimal for $r=2$ under one-shot linear schemes. For $K=Q=6$ and $r=2$, the scheme in \cite{LCW} wants the number of input files to be  $N=45$ $ (N=\binom{K-r-1}{r-1}\binom{K}{r})$.
\subsection{Wireless MapReduce Arrays from Extended Placement \\Delivery Arrays (EPDAs) in \cite{NPR,YWCQC}}
The EPDAs are introduced for multi-antenna coded caching \cite{NPR}, \cite{YWCQC}, and are defined as follows: For positive integers $K, r<K, N, Z<N, S$, and $g$, an $N\times K$ array consisting of $S$ integers and $\star$ is said to be a $g-$regular $(K,r,N,Z,S)$ EPDA if it satisfies the following conditions: (i) the symbol $\star$ appears $Z$ times in each column, (ii) every integer in $[S]$ occurs $g$ times in the array and no integer appears more than once in any column, and (iii) condition $A3$. 

The wireless MapReduce arrays can be obtained as a special case of the EPDAs as given below. 
\begin{lem}
 Any $2r-$regular $(K,r,N,Z,S)$ EPDA with $r=\frac{KZ}{N}$ is a $(K,N,r,S)$ wireless MapReduce array if each row contains exactly $r$ $\star$'s.	
 \label{lem2}
\end{lem}
In \cite{YWCQC}, it is shown that for any positive integers $K,t<K$, and $t\leq \Gamma <K$, there exists a $(t+\Gamma)-$regular $(K,\Gamma,N,tN/K,S)$ EPDA with $N=K(t+\Gamma)/(\text{gcd}(K,t,\Gamma))^2$ and $S=N(K-t)/(t+\Gamma)$ (from the multi-antenna coded caching scheme in \cite{SPSET}). Note that the resulting EPDA will have exactly $t$ $\star$'s in each row. By setting $t=\Gamma=r$, we have the following corollary from Lemma~\ref{lem2}.
\begin{cor}
  For any given positive integers $K$ and $r\leq K/2$, there exists a $(K,N,r,S)$ wireless MapReduce array with $N=2Kr/(\text{gcd}(K,r))^2$ and $S=N(K-r)/2r$. The scheme resulting from the above array achieves the normalized delivery time $L=\frac{1}{2r}\left(1-\frac{r}{K}\right)$, which is optimal under one-shot linear schemes, for the computation load $r$.
\end{cor}

\section{Conclusion}
In this work, we proposed wireless MapReduce distributed computing schemes that achieve optimal normalized delivery time, under one-shot linear schemes, with a smaller requirement on the number of input files than the existing schemes. The schemes are obtained by designing an array that jointly represents the Map, Shuffle, and Reduce phases. We also showed that the array structures known for multi-antenna coded caching schemes (EPDAs) can be used to obtain wireless MapReduce schemes.


\end{document}